\documentclass[journal,onecolumn,12pt,twoside]{IEEEtran}
\usepackage{amsfonts,color,pslatex}
\usepackage{amssymb,amsmath,latexsym}

\usepackage[para]{threeparttable}

\newcommand{\tr}{{\rm Tr}}
\newcommand{\gf}{{\rm GF}}

\newcommand{\C}{{\mathcal C}}

\newcommand{\wt}{{\mathrm{wt}}}

\newcommand{\bc}{{\mathbf{c}}}
\newcommand{\ttC}{{\mathtt C}}

\newtheorem{theorem}{Theorem}[section]
\newtheorem{lemma}[theorem]{Lemma}

\newtheorem{example}[theorem]{Example}

\begin{document}

\title{Five Families of Three-Weight Ternary Cyclic Codes and Their Duals\thanks{C. Ding's research was supported by
The Hong Kong Research Grants Council, Proj. No. 600812.
Y. Gao's research was partially supported by the National Natural
Science Foundation of China under Grant No. 11101019.
Z. Zhou's research was supported by
the Natural Science Foundation of China, Proj. No. 61201243, and also The Hong Kong Research
Grants Council, Proj. No. 600812. }
}

\author{Cunsheng~Ding, Ying~Gao,
and~Zhengchun Zhou   
\thanks{C. Ding is with the Department of Computer Science
                                                  and Engineering, The Hong Kong University of Science and Technology,
                                                  Clear Water Bay, Kowloon, Hong Kong, China (email: cding@ust.hk).}
\thanks{Y. Gao is with the School of Mathematics and Systems Science, Beijing University of Aeronautics and Astronautics, Beijing, China
                                        (email: gaoyingbuaa@gmail.com).}
\thanks{Z. Zhou is with the School of Mathematics, Southwest Jiaotong University,
Chengdu, 610031, China (email: zzc@home.swjtu.edu.cn).}
}

\maketitle

\begin{abstract}
As a subclass of linear codes, cyclic codes have applications in consumer electronics,
data storage systems, and communication systems as they have efficient encoding and
decoding algorithms. In this paper, five families of three-weight ternary cyclic codes whose
duals have two zeros are presented. The weight distributions of the five families of cyclic
codes are settled. The duals of two families of the cyclic codes are optimal.
\end{abstract}


\section{Introduction}

Let $p$ be a prime.
An $[n,\kappa, d]$ linear code over $\gf(p)$ is a $\kappa$-dimensional subspace of $\gf(p)^n$
with minimum nonzero (Hamming) weight $d$.
A linear $[n,\kappa]$ code $\C$ over $\gf(p)$ is called {\em cyclic} if
$(c_0,c_1, \cdots, c_{n-1}) \in \C$ implies $(c_{n-1}, c_0, c_1, \cdots, c_{n-2})
\in \C$.
By identifying any vector $(c_0,c_1, \cdots, c_{n-1}) \in \gf(p)^n$
with
$$
c_0+c_1x+c_2x^2+ \cdots + c_{n-1}x^{n-1} \in \gf(p)[x]/(x^n-1),
$$
any linear code $\C$ of length $n$ over $\gf(p)$ corresponds to a subset of the quotient ring
$\gf(p)[x]/(x^n-1)$.
A linear code $\C$ is cyclic if and only if the corresponding subset in $\gf(p)[x]/(x^n-1)$
is an ideal of the ring $\gf(p)[x]/(x^n-1)$.

It is well known that every ideal of $\gf(p)[x]/(x^n-1)$ is principal. Let $\C=\langle g(x) \rangle$
be a cyclic code, where $g(x)$ is monic and has the smallest degree among all the
generators of $\C$. Then $g(x)$ is unique and called the {\em generator polynomial,}
and $h(x)=(x^n-1)/g(x)$ is referred to as the {\em parity-check} polynomial of $\C$.
If the parity check polynomial $h(x)$ of a code $\C$ of length $n$ over $\gf(p)$ is the
product of $t$ distinct irreducible polynomials over $\gf(p)$, we say that the dual code
$\C^\perp$ has $t$ zeros.

Let $A_i$ denote the number of codewords with Hamming weight $i$ in a code
$\C$ of length $n$. The {\em weight enumerator} of $\C$ is defined by
$
1+A_1y+A_2y^2+ \cdots + A_ny^n.
$
The {\em weight distribution} $\{A_0,A_1,\ldots,A_n\}$ is an important research topic in coding theory.
First, it contains crucial information as to estimate the error correcting capability and the probability of
error detection and correction with respect to some algorithms \cite{Klov}. Second, due to rich algebraic
structures of cyclic codes, the weight distribution is often related to interesting and challenging problems
in number theory.

As a subclass of linear codes, cyclic codes have been widely used in consumer electronics, data
transmission technologies,
broadcast systems, and computer applications as they have efficient encoding and decoding
algorithms. Cyclic codes with a few weights are of special interest in authentication codes as
certain parameters of the authentication codes constructed from these cyclic codes are easy
to compute  \cite{DW05}, and in secret sharing schemes as the access structures of such
secret sharing schemes derived from such cyclic code are easy to determine \cite{CDY05,DS06,YD06}.
Cyclic codes with a few weights are also of special interest in designing frequency hopping sequences
\cite{DFFJ,DYT}. Three-weight cyclic codes have also applications in association schemes \cite{CG84}.
These are some of the motivations of studying cyclic codes with a few weights.

In this paper, five families of three-weight ternary cyclic codes whose duals have two zeros are
presented. The weight distributions of the five families of cyclic codes are settled. The duals of
two families of the cyclic codes proposed in this paper are optimal ternary codes. As a byproduct,
three new decimation values of maximum-length sequences giving only three correlation values
are also obtained in this paper.

This paper is organized as follows. Section \ref{sec-notations} fixes some notations for this paper.
Section \ref{sec-codes2zeros} defines cyclic codes over $\gf(p)$ whose duals have two zeros.
Section \ref{sec-prelim} presents two lemmas that will be needed in the sequel.
Section \ref{sec-7const} defines two families of cyclic codes and determines their weight distributions.
Section \ref{sec-2nd3wtc} describes three families cyclic codes and determines their weight distributions.
Section \ref{sec-fin} makes  concluding remarks.

\section{Some notations fixed throughout this paper}\label{sec-notations}

Throughout this paper, we adopt the following notations unless otherwise stated:
\begin{itemize}
\item $p$ is a prime and
          $q=p^m$, where $m$ is a positive integer.
\item $n=q-1$, which is  the length of a cyclic code over $\gf(p)$.
\item $\tr_{1}^{j}(x)$ is the trace function from $\gf(p^j)$ to $\gf(p)$ for any positive integer $j$.
\item $\chi$ is the canonical additive character on $\gf(q)$,  i.e., $\chi(x)=e^{2\pi \sqrt{-1}  \tr_1^{m}(x)/ p}$
         for any $x \in \gf(q)$.

\item $\ttC_a$ denotes the $p$-cyclotomic coset modulo $n$ containing $a$, where $a$ is any integer with
         $0 \le a \le r-2$, and $\ell_a:=|\ttC_a|$ denote the size of the cyclotomic coset $\ttC_a$.
\item By the Database we mean the collection of the tables of best linear codes known maintained by
         Markus Grassl at http://www.codetables.de/.
\end{itemize}

\section{Cyclic codes whose duals have two zeros}\label{sec-codes2zeros}

Given a positive integer $m$, recall that $q=p^m$ and $n=q-1$ throughout this paper.
Let $\alpha$ be a generator of the multiplicative group $\gf(q)^*$. For any $0 \le a \le q-2$,
denote by $m_a(x)$ the minimal polynomial of $\alpha^{-a}$ over $\gf(p)$.

Let $0 \le u \le q-2$ and $0 \le v \le q-2$ be any two integers such that $\ttC_u \cap \ttC_v = \emptyset$.
Let $\C_{(u,v,p,m)}$ be the cyclic code over $\gf(p)$ with length $n$ whose codewords are given by
\begin{eqnarray} \label{eqn-def-cab}
\bc(a,b)=(c_0,c_1,\ldots,c_{n-1}), \quad  \forall \, (a,b) \in \gf(p^{\ell_u})\times \gf(p^{\ell_v}),
\end{eqnarray}
where
\[c_i=\tr^{\ell_u}_{1}\left(a \alpha^{iu}\right)+\tr^{\ell_v}_{1}\left(b \alpha^{iv}\right), \quad 0 \le i \le n-1.\]
By Delsarte's Theorem, the code
$\C_{(u,v,p,m)}$ has parity-check polynomial $m_u(x)m_v(x)$ and dimension $\ell_u+\ell_v$.

In terms of exponential sums, the
Hamming weight $\wt(\bc(a,b))$ of the codeword $\bc(a,b)$ of (\ref{eqn-def-cab}) in $\C_{(u,v,p,m)}$ is given by
\begin{equation}\label{eqn-wtformula}
\wt(\bc(a,b))=(p-1)p^{m-1} -\frac{1}{p} \sum_{y \in \gf(p)^*} T_{v}(ya,yb)
\end{equation}
where
\begin{eqnarray}\label{eqn-T-sum}
T_{(u,v)}(a,b)=\sum_{x \in \gf(q)} \chi(ax^u+bx^v)
\end{eqnarray}
for each $(a, b) \in \gf(q)^2$. Throughout this section,  the function $T_{(u,v)}(a,b)$ is always defined as in (\ref{eqn-T-sum})
for any given $u$ and $v$.

The following lemma is an extension of Lemma 6.1 in \cite{ZD2}, and will be frequently used in the sequel when we determine the weight distributions of the five families of cyclic codes.

\begin{lemma}\label{Lemma-key}
Let $s$ be any integer with $\gcd(s,q-1)=2$. Then
\begin{eqnarray*}
T_{(u,v)}(a,b)=\frac{1}{2}\left(\sum_{x\in \gf(q)}\chi(ax^{su}+bx^{sv})+\sum_{x\in \gf(q)}\chi(a\lambda^u x^{su}+b\lambda^{v} x^{sv})           \right)
\end{eqnarray*}
where $\lambda$ is any fixed nonsquare in $\gf(q)^*$.
\end{lemma}

\begin{proof}
Let $C_0^{(2,q)}$
denote the set of all
nonzero squares in $\gf(q)$.
Then
\begin{eqnarray}\label{eqn-T-ab-2}
T_{(u,v)}(a,b)=1+\sum_{x\in C_0^{(2,q)}}\chi(ax^u+bx^{v})+\sum_{x\in C_0^{(2,q)}}\chi(a\lambda^u x^u+b\lambda^vx^v).
\end{eqnarray}
Note that $\gcd(q-1,s)=2$. When $x$ runs through $\gf(q)$,  $x^s$ runs twice through the
nonzero squares in $\gf(q)$ and takes on the value $0$ once. Similarly,
$\lambda x^s$ runs twice through all the
nonsquares in $\gf(q)$ and takes on the value $0$ once. The conclusion
then follows directly from (\ref{eqn-T-ab-2}) and the discussions above.
\end{proof}

There are a lot of references on the codes $\C_{(u,v,p,m)}$ (see for example
\cite{DL2,FL07,DL1,TFeng,Mc04,veg,veg2,WT,X,X2}).
This family of cyclic codes $\C_{(u,v,p,m)}$ may have many nonzero weights. It is obvious that $\C_{(u,v,p,m)}$
cannot be a constant-weight code as its parity-check polynomial has two zeros. When $q=2$, the codes
$\C_{(u,v,p,m)}$ cannot be a two-weight code if $\ttC_u \cap \ttC_v = \emptyset$ and $\ell_u>1$
and $\ell_v >1$. When $q$ is odd, $\C_{(u,v,p,m)}$ could be a two-weight code.

A number of three-weight nonbinary cyclic codes $\C_{(u,v,p,m)}$ have been constructed (see for example
\cite{CDY05,Choi,FL07,LuoFeng2,Xia,YCD}). In this paper, we present five families of three-weight ternary
cyclic codes $\C_{(u,v,p,m)}$, determine their weight distributions and study their duals.

\section{Two families of three-weight ternary cyclic codes and their weight enumerators}\label{sec-7const}

In this section, we propose two families of three-weight cyclic codes $\C_{(u, v, 3, m)}$ over $\gf(3)$ where $u$ and $v$
are some integers with $\ttC_u \cap \ttC_v = \emptyset$ and $(\ell_u, \ell_v)=(m,m)$. It is obvious that the code  $\C_{(u, v, 3, m)}$ has length $3^m-1$ and dimension $2m$ under these assumptions.

\subsection{Some auxuliary results}\label{sec-prelim}

In this subsection, we introduce a lemma on exponential sums over finite fields and a lemma
regarding the dual code $\C_{(1,v,3,m)}^\perp$ of the code $\C_{(1,v,3,m)}$. Recall that
$\chi$ is the canonical additive character of $\gf(q)$ defined in Section \ref{sec-notations}.

The following lemma will be employed in the sequel, and was proved in \cite{ZD2} with the help
of some results from \cite{YCD,LuoFeng2}.

\begin{lemma}\label{Lemma-Quadratic-expo-1}
Let $m$ be odd and $h$ be an integer  with $\gcd(m,h)=1$. Define
\begin{eqnarray*}
R(a,b)=\sum_{x\in \gf(q)}\chi(ax^{p^h+1}+bx^2).
\end{eqnarray*}
Then, as $(a,b)$ runs through $\gf(q)^2$,  the values of the sum
$$\sum_{y\in \gf(p)^*}(R(ya,yb)+R(-ya,yb))$$
have the following distribution
\begin{eqnarray*}
\begin{array}{cc}
\textrm{Value}&  ~~~     \textrm{Frequency}\\
2(p-1)p^{m}  &     ~~~    1\\
(p-1)p^{(m+1)/ 2}&  ~~~   (p^{m-1}+p^{(m-1)/ 2})(p^m-1)\\
0&                    (p^m-2p^{m-1}+1)(p^m-1)\\
-(p-1)p^{(m+1)/ 2}&  ~~  ~(p^{m-1}-p^{(m-1)/ 2})(p^m-1).
\end{array}
\end{eqnarray*}
\end{lemma}

\begin{table}[ht]
\caption{Weight distribution I}\label{tab-CGAPN}
\centering
\begin{tabular}{|l|l|}
\hline
\hline
Weight $w$    & No. of codewords $A_w$  \\ \hline
  $0$ & $1$ \\
  \hline
 $(p-1)p^{m-1}-\frac{p-1}{2} p^{(m-1)/2 }$ & $(p^m-1)(p^{m-1}+p^{(m-1)/ 2})$\\
\hline
$(p-1)p^{m-1}$ & $(p^m-1)(p^m-2p^{m-1}+1)$\\
\hline
$(p-1)p^{m-1}+\frac{p-1}{2} p^{(m-1)/ 2}$ & $(p^m-1)(p^{m-1}-p^{(m-1)/ 2})$\\
\hline
\hline
\end{tabular}
\begin{tablenotes}
\end{tablenotes}
\end{table}

The following lemma is proved in \cite{DH12} and will be employed in the sequel.

\begin{lemma}\label{lem-DH12}
Let $v \not\in \ttC_1$ and $\ell_v=|\ttC_v|=m$. Then the dual $\C_{(1,v,3,m)}^\perp$ of the
ternary cyclic code $\C_{(1,v,3,m)}$
has parameters $[3^m-1, 3^m-1-2m, 4]$ if and only if the following conditions are satisfied:
\begin{itemize}
\item[C1:] $v$ is even;
\item[C2:] the equation $(-x-1)^v+x^v+1=0$ has the only solution $x=1$ in $\gf(q)^*$; and
\item[C3:] the equation $(x+1)^v-x^v-1=0$ has the only solution $x=0$ in $\gf(q)$.
\end{itemize}
\end{lemma}

\subsection{The first family of three-weight cyclic codes}

In this subsection, we study the cyclic codes $\C_{(1, v, p, m)}$,
where  $m$ is odd, $p=3$, and $v=({3^{m+1}-1})/4$.
The parameters of the codes are described in the following theorem.

\begin{theorem}\label{thm-class-1}
Let $m$ be odd, $p=3$, and $v=({3^{m+1}-1})/4$. Then $\C_{(1, v, p, m)}$
is a $[p^m-1, 2m]$ cyclic code over $\gf(p)$ with the weight distribution in Table
\ref{tab-CGAPN}.
\end{theorem}

\begin{proof}
Let  $h=1$ and $s=3^{h}+1$.  Then $\gcd(s, 3^m-1)=2$ since $m$ is odd. It is easy to check that $sv\equiv 2~(\bmod~3^m-1)$.
 Noticing that $v$ is even and $-1$ is a nonsquare in $\gf(q)$.
By Lemma \ref{Lemma-key}, we have
\begin{eqnarray*}
T_{(1,v)}(a,b)
&=& \frac{1}{2} \left( R_{(1,v)}(a,b) + R_{(1,v)}(-a,b)  \right)
\end{eqnarray*}
where
$$
R_{(1,v)}(a,b)=\sum_{x \in \gf(q)} \chi(ax^{3^{h}+1}+bx^2).
$$
It then follows from (\ref{eqn-wtformula}) that
\begin{eqnarray}\label{eqn-APNcase1}
\wt(\bc(a,b))= 2\times 3^{m-1} - \frac{1}{6} \sum_{y\in \gf(3)^*}\left(R_{(1,v)}(y a,y b)+R_{(1,v)}(-y a,y b)\right).
\end{eqnarray}
Note that $\gcd(m,h)=\gcd(m,1)=1$, the weight distribution of the code $\C_{(1, v, 3, m)}$ then
follows  from Equation (\ref{eqn-APNcase1}) and Lemma \ref{Lemma-Quadratic-expo-1}.
\end{proof}

\begin{example}
Let $p=3$ and $m=3$. Let $\alpha$ be the generator of $\gf(p^m)^*$ with $\alpha^3+2\alpha+1=0$.
Then $v=20$ and $\C_{(1,v,p,m)}$ is a $[26,6,15]$ code over $\gf(3)$ with parity-check polynomial
$x^6 + 2x^3 + 2x^2 + x + 2$ and weight enumerator
$1+ 312 y^{15}+260 y^{18}+156 y^{21}.$
It has the same parameters as the best known cyclic codes
in the Database, and is optimal.
\end{example}

\begin{example}
Let $p=3$ and $m=5$. Let $\alpha$ be the generator of $\gf(p^m)^*$ with $\alpha^5 + 2\alpha + 1=0$.
 Then $v=182$ and $\C_{(1,v,p,m)}$  is a $[242,10,153]$ code over $\gf(3)$ with parity-check polynomial
 $x^{10} + 2x^9 + 2x^8 + 2x^7 + 2x^5 + x^4 + 2x^3 + x^2 + x + 2$
 and weight enumerator
\begin{eqnarray*}
1+ 21780 y^{153}+19844 y^{162}+17424 y^{171}.
\end{eqnarray*}
It has the same parameters as the best known cyclic codes
in the Database. It is optimal or almost optimal since the upper bound on the minimal distance of  any ternary linear code with
length $242$ and dimension $6$ is $154$.
\end{example}

The following theorem describes the parameters of the dual code $\C_{(1,v,3,m)}^\perp$ of the code
$\C_{(1,v,3,m)}$ in Theorem \ref{thm-class-1}.

\begin{theorem}\label{thm-class-1d}
Let $m$ be odd, $p=3$, and $v=({3^{m+1}-1})/4$. Then $\C_{(1, v, p, m)}^\perp$
is a $[3^m-1, 3^m-1-2m, 4]$ cyclic code over $\gf(3)$.
\end{theorem}

\begin{proof}
The dimension of $\C_{(1, v, p, m)}^\perp$ follows from that of $\C_{(1, v, p, m)}$. So we need to prove
that the minimum distance $d^\perp$ of $\C_{(1, v, p, m)}^\perp$ is 4. To this end, we prove
that the three conditions in Lemma \ref{lem-DH12} are met.

Obviously, $v$ is even. So Condition C1 in Lemma \ref{lem-DH12} is satisfied. 

We now consider Condition
C2 in Lemma \ref{lem-DH12}  and study the solutions $x \in \gf(q)$ of the following equation
\begin{equation}\label{eqn-c1stc1.1}
(x+1)^v + x^v = -1.
\end{equation}
It is clear that Equation (\ref{eqn-c1stc1.1}) is equivalent to the following
\begin{equation}\label{eqn-c1stc1.2}
y-x=1, \ \ y^v + x^v =- 1.
\end{equation}
Note that $-1$ is a nonsquare in $\gf(q)$ and $\gcd(4, q-1)=2$ as $m$ is odd. We now consider the solutions
$(x, y) \in \gf(q)^2$ of (\ref{eqn-c1stc1.2}) by distinguishing among the following four cases.

\subsubsection*{Case 1, $x=x_1^4$ and $y=y_1^4$ for some $(x_1, y_1) \in \gf(q)^2$}

In this case, $(x_1, y_1)$ is a solution of
\begin{equation}\label{eqn-c1stc1.3}
y_1^4-x_1^4=1, \ \ y_1^2 + x_1^2 = -1.
\end{equation}
It then follows that
$$
y_1^2+x_1^2=-1, \ \ y_1^2 - x_1^2 = -1.
$$
Hence $y_1^2=-1$, which is impossible as $-1$ is not a square in $\gf(3^m)$.

\subsubsection*{Case 2, $x=-x_1^4$ and $y=-y_1^4$ for some $(x_1, y_1) \in \gf(q)^2$}

In this case, $(x_1, y_1)$ is a solution of
\begin{equation}\label{eqn-c1stc1.4}
x_1^4-y_1^4=1, \ \ y_1^2 + x_1^2 = -1.
\end{equation}
It then follows that
$$
y_1^2+x_1^2=-1, \ \ y_1^2 - x_1^2 = 1.
$$
Hence $y_1=0$ and $-1=x_1^2$, which is impossible as $-1$ is not a square in $\gf(q)$.

\subsubsection*{Case 3, $x=x_1^4$ and $y=-y_1^4$ for some $(x_1, y_1) \in \gf(q)^2$}

In this case, $(x_1, y_1)$ is a solution of
\begin{equation}\label{eqn-c1stc1.5}
y_1^4+x_1^4=-1, \ \ y_1^2 + x_1^2 = -1.
\end{equation}
It then follows that
$$
1=(y_1^2 + x_1^2)^2=-1+2(x_1y_1)^2
$$
and $1=(x_1y_1)^2$. Hence $x_1^2$ and $y_1^2$ are the solutions of $z^2-2z+1=0$. Thus,
$x_1^2=y_1^2=1$.  Whence $x=x_1^4=1$.

\subsubsection*{Case 4, $x=-x_1^4$ and $y=y_1^4$ for some $(x_1, y_1) \in \gf(q)^2$}

In this case, $(x_1, y_1)$ is a solution of
\begin{equation}\label{eqn-c1stc1.6}
y_1^4+x_1^4=1, \ \ y_1^2 + x_1^2 = -1.
\end{equation}
It then follows that
$$
1=(y_1^2 + x_1^2)^2=1+2(x_1y_1)^2.
$$
Hence $0=(x_1y_1)^2$. Therefore $x_1=0$ or $y_1=0$. Hence $-1 =y_1^2$ or $-1 =x_1^2$,
which are impossible.

Summarizing the four cases above, we proved that Condition C3  in Lemma \ref{lem-DH12}
is satisfied.

One can similarly prove that Condition C3  of Lemma \ref{lem-DH12} is also met.  

The desired conclusion of the parameters of the code then follows from  Lemma \ref{lem-DH12}.
\end{proof}

Note that the codes $\C_{(1,v,p,m)}^\perp$ of Theorem \ref{thm-class-1d} are optimal in the sense that
the minimum distance of any ternary linear code of length $3^m-1$ and dimension $3^m-1-2m$ is at
most 4 due to the sphere-packing bound.

\begin{example}
Let $p=3$ and $m=3$. Let $\alpha$ be the generator of $\gf(p^m)^*$ with $\alpha^3+2\alpha+1=0$.
Then $v=20$ and $\C_{(1,v,p,m)}^\perp$ is a $[26,20,4]$ code over $\gf(3)$ with generator polynomial
$x^6 + 2x^3 + 2x^2 + x + 2$.
\end{example}

\subsection{The second family of three-weight cyclic codes}

In this subsection, we analyze the cyclic codes $\C_{(1, v, p, m)}$, where $m \equiv 7 \pmod{8}$, $p=3$,
and $v=\left(3^{(m+1)/8}-1\right)\left(3^{(m+1)/4}+1\right)\left(3^{(m+1)/2}+1\right)$.
The parameters of the codes are described in the following theorem.

\begin{theorem}\label{thm-class-6}
Let $m \equiv 7 \pmod{8}$, $p=3$ and $v=(3^{(m+1)/8}-1)(3^{(m+1)/4}+1)(3^{(m+1)/2}+1)$.
Then $\C_{(1, v, 3, m)}$
is a $[3^m-1, 2m]$ cyclic code over $\gf(3)$ with the weight distribution in Table
\ref{tab-CGAPN}.
\end{theorem}

\begin{proof}
Let $h=(m+1)/8$ and $s=3^h+1$.  Since $m \equiv 7 \pmod{8}$, $\gcd(s, 3^m-1)=2$ and $v$ is even.
It is easy to check that $sv\equiv 2~(\bmod~3^m-1)$.
Select $\lambda=-1$ as a nonsquare in $\gf(p^m)$.
Applying Lemma \ref{Lemma-key}, we have
\begin{eqnarray}\label{eqn-T-six-class}
T_{(1,v)}(a,b)
&=& \frac{1}{2} \left( R_{(1,v)}(a,b) + R_{(1,v)}(-a,b)  \right)
\end{eqnarray}
where
$$
R_{(1,v)}(a,b)=\sum_{x \in \gf(q)} \chi(ax^{3^{h}+1}+bx^2).
$$
It then follows from (\ref{eqn-wtformula}) and (\ref{eqn-T-six-class}) that
\begin{eqnarray}\label{eqn-APNcase6}
\wt(\bc(a,b))= 2\times 3^{m-1} - \frac{1}{6} \sum_{y\in \gf(3)^*}\left(R_{(1,v)}(y a,y b)+R_{(1,v)}(-y a,y b)\right).
\end{eqnarray}
The weight distribution of the code $\C_{(1, v, 3, m)}$ then
follows from (\ref{eqn-APNcase6}) and Lemma \ref{Lemma-Quadratic-expo-1}.
\end{proof}

\begin{example}
Let $p=3$ and $m=7$. Let $\alpha$ be a generator of $\gf(q)^*$ with  $\alpha^7 + 2\alpha^2 + 1=0$.
Then $v=1640$ and  $\C_{(1,v,p,m)}$ is a $[2186,14,1431]$ code over $\gf(3)$ with
parity-check polynomial
$$
x^{14} + 2x^{13} + x^{12} + x^{11} + x^9 + 2x^8 + 2x^7 + x^6 + 2x^3 + x^2 + x  + 2
$$
and weight enumerator
\begin{eqnarray*}
1+ 1652616 y^{1431}+1595780 y^{1458}+1534572 y^{1485}.
\end{eqnarray*}
\end{example}

Now we determine the parameters of the dual code $\C_{(1,v,3,m)}^\perp$ of
$\C_{(1,v,3,m)}$ in Theorem \ref{thm-class-6}. To this end, we need the following lemmas whose
proofs are omitted.

\begin{lemma}\label{lem-gamma}
Let $m \equiv 7 \pmod{8}$ and $q=3^m$. Define $h=(m+1)/8$ and $s=3^h+1$. Then all solutions
$(x_1, y_1) \in \gf(q) \times \gf(q)$ of the equation $y_1^2-x_1^2=1$ can be expressed as
\begin{eqnarray}\label{eqn-4all1}
x_1=\theta - \theta^{-1}, \ \  y_1=-\theta -  \theta^{-1}
\end{eqnarray}
where $\theta \in \gf(q)^*$.
In addition, we have
\begin{eqnarray}\label{eqn-4all2}
\left\{
\begin{array}{lll}
x_1^s&=&\theta^s + \theta^{-s} - ( \theta^{3^{(m+1)/8}-1} + \theta^{1-3^{(m+1)/8}} ) \\
y_1^s&=&\theta^s + \theta^{-s} + ( \theta^{3^{(m+1)/8}-1} + \theta^{1-3^{(m+1)/8}} ).
\end{array}
\right.
\end{eqnarray}
\end{lemma}

\begin{lemma}\label{lem-beta}
Let $m \equiv 7 \pmod{8}$. Define $h=(m+1)/8$ and $s=3^h+1$.
If $m \equiv 7 \pmod{16}$, then
\begin{eqnarray*}
\left\{ \begin{array}{l}
s \equiv 4 \pmod{8} \\
\gcd(s, q^2-1)=4 \\
\gcd(3^h-1, q^2-1)=2.
\end{array}
\right.
\end{eqnarray*}
If $m \equiv -1 \pmod{16}$, then
\begin{eqnarray*}
\left\{ \begin{array}{l}
s \equiv 2 \pmod{8} \\
\gcd(s, q^2-1)=2 \\
\gcd(3^h-1, q^2-1)=8.
\end{array}
\right.
\end{eqnarray*}

\end{lemma}

\begin{lemma}\label{lem-alpha}
Let $m$ be odd and $q=3^m$.
Let $\gamma$ be a generator of $\gf(q^2)$ and let $\epsilon=\gamma^{(q^2-1)/4}$. Then
All solutions of $(x_1, y_1) \in \gf(q^2) \times \gf(q^2)$ of the equation $x_1^2+y_1^2=-1$
can be expressed as
\begin{eqnarray}\label{eqn-bsolu}
x_1=\epsilon (\theta + \theta^{-1}), \ y_1=\theta - \theta^{-1}
\end{eqnarray}
for some $\theta \in \gf(q^2)^*$, where $\epsilon^2=-1$.

Furthermore,
\begin{eqnarray*}
x_1^s &=&  (\theta^s+\theta^{-s} + \theta^{3^h-1} + \theta^{1-3^h} ) \epsilon^s \\
y_1^s &=&  \theta^s+\theta^{-s} - (\theta^{3^h-1} + \theta^{1-3^h}).
\end{eqnarray*}
\end{lemma}

The following theorem describes the parameters of the dual code $\C_{(1,v,3,m)}^\perp$ of
$\C_{(1,v,3,m)}$ in Theorem \ref{thm-class-6}.

\begin{theorem}\label{thm-class-6d}
Let $m \equiv 7 \pmod{8}$, $p=3$, and $v=(3^{(m+1)/8}-1)(3^{(m+1)/4}+1)(3^{(m+1)/2}+1)$. Then $\C_{(1, v, p, m)}^\perp$
is a $[3^m-1, 3^m-1-2m, 4]$ cyclic code over $\gf(3)$.
\end{theorem}

\begin{proof}
The dimension of $\C_{(1, v, p, m)}^\perp$ follows from that of $\C_{(1, v, p, m)}$. So we need to prove
that the minimum distance $d^\perp$ of $\C_{(1, v, p, m)}^\perp$ is 4. To this end, we prove
that the three conditions in Lemma \ref{lem-DH12} are met.

Define $h=(m+1)/8$ and $s=3^h+1$. It is easily seen that $\gcd(s, q-1)=2$ and $sv \equiv 2 \pmod{q-1}$.
Note that $-1$ is a nonsquare in $\gf(q)$  as $m$ is odd.

Obviously, $v$ is even. So Condition C1 in Lemma \ref{lem-DH12} is satisfied. We now consider Condition
C3, and study the solutions $x \in \gf(q)$ of the following equation
\begin{equation}\label{eqn-1stc6.1}
(x+1)^v - x^v = 1.
\end{equation}
It is clear that Equation (\ref{eqn-1stc6.1}) is equivalent to the following system of equations
\begin{equation}\label{eqn-1stc6.2}
y-x=1, \ \ y^v - x^v = 1.
\end{equation}

We now consider the solutions
$(x, y) \in \gf(q)^2$ of (\ref{eqn-1stc6.2}) by distinguishing among the following four cases.

\subsubsection*{Case 1, $x=x_1^s$ and $y=y_1^s$ for some $(x_1, y_1) \in \gf(q)^2$}

In this case, $(x_1, y_1)$ is a solution of
\begin{equation}\label{eqn-2stc2.3}
y_1^s-x_1^s=1, \ \ y_1^2 - x_1^2 = 1.
\end{equation}
Let $\theta=y_1-x_1$. Clearly, $\theta \ne 0$.
Then plugging the expressions of $(x_1, y_1)$ in Lemma \ref{lem-gamma} into
the first equation of (\ref{eqn-2stc2.3}) yields
\begin{eqnarray}\label{eqn-2stc2.5}
 \theta^{3^{(m+1)/8}-1} + \theta^{1-3^{(m+1)/8}} = -1.
\end{eqnarray}

Obviously, the equation
$
 z + z^{-1} = -1
$
has the unique solution $z=1$. In this case, it follows from (\ref{eqn-2stc2.5})
that $\theta^{3^{(m+1)/8}-1}=1$. Note that $\gcd(3^{(m+1)/8}-1, 3^m-1)=2$. Hence $\theta=\pm 1$.
It then follows from the first equation in (\ref{eqn-4all1}) that
$
x_1=\theta -\theta^{-1}=0.
$
Hence $x=x_1^s=0$.

\subsubsection*{Case 2, $x=-x_1^s$ and $y=-y_1^s$ for some $(x_1, y_1) \in \gf(q)^2$}

In this case, $(x_1, y_1)$ is a solution of
\begin{equation}\label{eqn-22stc2.3}
y_1^s-x_1^s=-1, \ \ y_1^2 - x_1^2 = 1.
\end{equation}
Let $\theta=y_1-x_1$. Clearly, $\theta \ne 0$.
Then plugging the expressions of $(x_1, y_1)$ in Lemma \ref{lem-gamma} into
 the first equation of (\ref{eqn-22stc2.3}) yields
\begin{eqnarray}\label{eqn-22stc2.5}
 \theta^{3^{(m+1)/8}-1} + \theta^{1-3^{(m+1)/8}} = 1.
\end{eqnarray}

The equation
$
 z + z^{-1} = 1
$
has the unique solution $z=-1$. It then follows from (\ref{eqn-22stc2.5}) that
$ \theta^{3^{(m+1)/8}-1}=-1$. This is impossible as $-1$ is not a square in $\gf(q)$.

\subsubsection*{Case 3, $x=-x_1^s$ and $y=y_1^s$ for some $(x_1, y_1) \in \gf(q)^2$}

In this case, $(x_1, y_1)$ is a solution of
\begin{equation}\label{eqn-322stc2.3}
y_1^s+x_1^s=1, \ \ y_1^2 - x_1^2 = 1.
\end{equation}
Let $\theta=y_1-x_1$. Clearly, $\theta \ne 0$.
Then plugging the expressions of $(x_1, y_1)$ in Lemma \ref{lem-gamma} into
 the first equation of (\ref{eqn-322stc2.3}) gives
$
2(\theta^s + \theta^{-s})=1.
$
Note that $\gcd(s, 3^m-1)=2$. We obtain that $\theta=\pm 1$. It then follows that
$
x_1=0
$
and $x=x_1^s=0$.

\subsubsection*{Case 4, $x=x_1^s$ and $y=-y_1^s$ for some $(x_1, y_1) \in \gf(q)^2$}

In this case, $(x_1, y_1)$ is a solution of
\begin{equation}\label{eqn-4322stc2.3}
y_1^s+x_1^s=-1, \ \ y_1^2 - x_1^2 = 1.
\end{equation}
Let $\theta=y_1-x_1$. Clearly, $\theta \ne 0$.
Then plugging the expressions of $(x_1, y_1)$ in Lemma \ref{lem-gamma} into
the first equation of (\ref{eqn-4322stc2.3}) gives
$$
\theta^s + \theta^{-s}=1,
$$
which is impposible, as $z+z^{-1}=1$ has no solution $x\in \gf(q)$.

Summarizing the conclusions in the four cases above, we proved that Condition C3 in Lemma \ref{lem-DH12}
is satisfied. 

One can similarly prove that Condition  C2 in Lemma \ref{lem-DH12} is satisfied. 

Finally, the desired conclusions on the parameters of this code follow from Lemma \ref{lem-DH12}.
\end{proof}

Note that the codes $\C_{(1,v,p,m)}^\perp$ of Theorem \ref{thm-class-6d} are optimal in the sense that
the minimum distance of any ternary linear code of length $3^m-1$ and dimension $3^m-1-2m$ is at
most 4 due to the sphere-packing bound.

\section{Three families of three-weight ternary cyclic codes and their weight enumerators}\label{sec-2nd3wtc}

In this section, we present three families of three-weight codes whose weight distributions are given
in Table \ref{tab-LuoFeng} and are different from the one in Table \ref{tab-CGAPN}. To this end, we need
the following lemma proved by Feng and Luo \cite{FL08}.

\begin{lemma}\label{lem-FL08}
Let $m\ge 3$ be odd, $p$ be an odd prime and let $h$ be a positive integer with $\gcd(m, h)=1$. Then the code
$\C_{(p^h+1, 2, p, m)}$ has dimension $2m$ and the weight distribution in Table \ref{tab-LuoFeng}.
\end{lemma}

\begin{table}[ht]
\caption{Weight distribution II}\label{tab-LuoFeng}
\centering
\begin{tabular}{|l|l|}
\hline
\hline
Weight $w$    & No. of codewords $A_w$  \\ \hline
  $0$ & $1$ \\
  \hline
 $(p-1)(p^{m-1}-p^{(m-1)/2 })$ & $\frac{1}{2}(p^m-1)(p^{m-1}+p^{(m-1)/ 2})$\\
\hline
$(p-1)p^{m-1}$ & $(p^m-1)(p^m-p^{m-1}+1)$\\
\hline
$(p-1)(p^{m-1}+p^{(m-1)/2 })$ & $\frac{1}{2}(p^m-1)(p^{m-1}-p^{(m-1)/ 2})$\\
\hline
\hline
\end{tabular}
\begin{tablenotes}
\end{tablenotes}
\end{table}

\begin{theorem}\label{thm-class-d3}
Let $m \ge 3$ be odd and $p=3$. Then $\C_{(1, v, p, m)}$
is a $[p^m-1, 2m]$ ternary cyclic code  with the weight distribution in Table \ref{tab-LuoFeng} if
\begin{itemize}
\item $v=(3^{m+1}-1)/(3^h+1)+(3^m-1)/2$, where $(m+1)/h$ is even; or
\item $v=(3^{(m+1)/8}-1)(3^{(m+1)/4}+1)(3^{(m+1)/2}+1) + ({3^m-1})/{2}$,
          where $m \equiv 7 \pmod{8}$; or
 \item $v=(3^{(m+1)/4}-1)(3^{(m+1)/2}+1) + ({3^m-1})/{2}$,
          where $m \equiv 3 \pmod{4}$.
\end{itemize}
\end{theorem}

\begin{proof}
We now prove the conclusion of this theorem for the first $v$.
Let $p=3$ and $q=p^m-1$. Define $\lambda=\alpha^{(q-1)/(p-1)}$, where $\alpha$ is a generator of
$\gf(q)^*$. Clearly $\lambda$ is a generator of $\gf(p)^*$. Since $m$ is odd, $\lambda$ is a nonsquare
in $\gf(q)$. It is easy to verify that $\lambda^v=\lambda$.

Let  $s=p^{h}+1$. Then $\gcd(s,p^m-1)=2$ since $m$ is odd.
It is easy to verify that $v$ is odd and $sv\equiv 2~(\bmod~q-1)$.
Applying Lemma \ref{Lemma-key}, we have
\begin{eqnarray*}
T_{(1,v)}(a,b)
&=& \frac{1}{2}\left(\sum_{x\in \gf(q)}\chi(ax^{s}+bx^{sv})+\sum_{x\in \gf(q)}\chi(a\lambda x^{s}+b\lambda^{v} x^{sv})           \right) \\
&=& \frac{1}{2}\left(\sum_{x\in \gf(q)}\chi(ax^{s}+bx^{sv})+\sum_{x\in \gf(q)}\chi(a\lambda x^{s}+b\lambda x^{sv})           \right) \\
&=& \frac{1}{2}\left(\sum_{x\in \gf(q)}\chi(ax^{p^h+1}+bx^{2})+\sum_{x\in \gf(q)}\chi(a\lambda x^{p^h+1}+b\lambda x^{2})           \right).
\end{eqnarray*}
Hence,
$$
\sum_{y \in \gf(p)^*} T_{(1,v)}(ay,by)= \sum_{y \in \gf(p)^*} \sum_{x\in \gf(q)}\chi(ayx^{p^h+1}+byx^{2}).
$$
It then follows from (\ref{eqn-wtformula}) that the Hamming weight of the codeword $\bc(a,b)$ in $\C_{(1, v, p, m)}$
is equal to that of  the codeword $\bc(a,b)$ in  $\C_{(p^h+1, 2, p, m)}$. Hence the two codes have the same
weight distribution. The desired conclusion on the weight distribution of the code $\C_{(1, v, p, m)}$ follows from
Lemma \ref{lem-FL08} as $\gcd(h, m)=1$.

The proof of the conclusion for $v=(3^{(m+1)/8}-1)(3^{(m+1)/4}+1)(3^{(m+1)/2}+1) + ({3^m-1})/{2}$ is similar to that for the first $v$ except that we need to set $h=(m+1)/8$ and is omitted.

The proof of the conclusion for $(3^{(m+1)/4}-1)(3^{(m+1)/2}+1) + ({3^m-1})/{2}$ is similar
to that for the first $v$ except that we need to set $h=(m+1)/4$ and is omitted.
\end{proof}

The dual codes  $\C_{(1, v, p, m)}^\perp$ of the codes $\C_{(1, v, p, m)}$ in Theorem \ref{thm-class-d3} has
parameters $[p^m-1, p^m-1-2m, 2]$. The minimum distance of  $\C_{(1, v, p, m)}^\perp$ is 2 as $v$ is odd.

\section{Summary and concluding remarks}\label{sec-fin}

In this paper, we presented five families of three-weight ternary cyclic codes and settled their weight distributions.
The duals of the first two families of ternary codes are optimal. The first two families of cyclic codes have
the weight distribution in Table \ref{tab-CGAPN}, while the last three families have the weight distribution of table
\ref{tab-LuoFeng}. It would be interesting to investigate the applications of these cyclic codes in authentication codes, secret sharing and frequency hopping sequences using the frameworks developed in \cite{CDY05,DFFJ,DS06,DW05,YD06}.

The key technique for settling the weight distributions of these cyclic codes in this paper is the application of the
noninvertible transformations  $x \mapsto x^s$ from $\gf(q)$ to $\gf(q)$ with $\gcd(s, q-1)=2$ and Lemma
\ref{Lemma-key}. With these innovations we were able to determine the weight distributions of the cyclic codes
with the help of known results on certain exponential sums.

Note that $\gcd(v, q-1)=1$ and $v-1 \equiv 0 \pmod{p-1}$ for all the $v$'s listed in Theorem \ref{thm-class-d3}.
It follows from the discussions in Section A2 in \cite{Katz} and the weight distribution of the code $\C_{(1, v, p, m)}$
of Theorem \ref{thm-class-d3} that the crosscorrelation function of any maximum-length sequence of period $q-1$
over $\gf(p)$ and its $v$-decimated version takes on only the following three  correlation values:
$$
-1-p^{(m+1)/2}, \ -1, -1+p^{(m+1)/2}.
$$
These three-level decimation values $v$ should be new and form another contribution of this paper to the theory of sequences. 

\section*{Acknowledgments}

The authors are very grateful to the reviewers and the Associate Editor, Prof. Ian F. Blake, for their comments and suggestions
that improved the presentation and quality of this paper.


\begin{thebibliography}{99}

\bibitem{CG84} A.R. Calderbank and J.M. Goethals, ``Three-weight codes and association schemes,''
{\em Philips J. Res.,}  vol. 39, pp. 143--152, 1984.

\bibitem{CDY05} C. Carlet, C. Ding, and J. Yuan, ``Linear codes from perfect nonlinear mappings and
their secret sharing schemes,'' {\em IEEE Trans. Inform. Theory,} vol. 51, no. 6, pp. 2089--2102, June 2005.

\bibitem{Choi}  S.-T. Choi, J.-Y. Kim, J.-S. No, and H. Chung, ``Weight distribution of some cyclic codes,''
in: {\em Proc. of the 2012 International Symposium on Information Theory,} IEEE Press, 2012,
pp. 2911--2913.

\bibitem{DFFJ} C. Ding, R. Fuji-Hara, Y. Fujiwara, M. Jimbo, and M.Mishima, ``Sets of frequency hopping
sequences: bounds and optimal constructions,'' {\em IEEE Trans. Inform. Theory}, vol. 55, no. 7,
pp. 3297--3304, July 2009.

\bibitem{DH12} C. Ding and T. Helleseth, ``Optimal ternary cyclic codes from monomials;"
{\em IEEE Trans. Inform. Theory,}  vol. 59, no. 9, pp. 5898--5904, Sept. 2013. 

\bibitem{DL2} C. Ding, Y. Liu, C. Ma, and L. Zeng, ``The weight distributions of the duals of cyclic codes
with two zeros,'' {\em IEEE Trans. Inform. Theory,} vol. 57, no. 12, pp. 8000--8006, Dec. 2011.

\bibitem{DYT} C. Ding, Y. Yang, and X. Tang, ``Optimal sets of frequency hopping sequences from linear
cyclic codes,'' {\em IEEE Trans. Inform. Theory,} vol. 56, no. 7,  pp. 3605--3612, July 2010.

\bibitem{DS06} C. Ding and A. Salomaa, ``Secret sharing schemes with nice access structures,''
{\em Fundamenta Informaticae,} vol. 71, nos. 1--2,  pp. 65--79, 2006.

\bibitem{DW05} C. Ding and X. Wang, ``A coding theory construction of new systematic authentication
codes,'' {\em Theoretical Computer Science,} vol. 330, no. 1,  pp. 81--99, 2005.

\bibitem{FL07} K. Feng and J. Luo, ``Value distribution of exponential sums from
perfect nonlinear functions and their applications," {\em IEEE Trans. Inform.
Theory,} vol. 53, no. 9, pp. 3035--3041, Sept. 2007.

\bibitem{FL08} K. Feng and J. Luo, ``Weight distribution of some reducible cyclic codes,"
{\em Finite Fields and Their Applications,} vol. 14, no. 2, pp. 390--409, April 2008.

\bibitem{TFeng} T. Feng, ``On cyclic codes of length $2^{2^r}-1$ with two zeros whose
dual codes have three weights," {\em Des. Codes Cryptogr.,}  vol. 62, no. 3, pp. 253--258,
Mar. 2012.

\bibitem{Katz} D.J. Katz, ``Weil sums of binomials, three-level cross-correlation, and a conjecture
of Helleseth," {\em J. Comb. Theory Ser. A,} vol. 119, pp. 1644--1659, 2012.

\bibitem{Klov} T. Kl{\o}ve, {\em Codes for Error Detection,} World Scientific, Singapore, 2007.

\bibitem{LuoFeng2} J. Luo and K. Feng, ``On the weight distributions of two classes
of cyclic codes," {\em IEEE Trans. Inform. Theory,}  54, no. 12, pp. 5332--5344, Dec. 2008.

\bibitem{DL1} C. Ma, L. Zeng, Y. Liu, D. Feng, and C. Ding, ``The weight enumerator of a class
of cyclic codes," {\em IEEE Trans. Inform. Theory}, vol. 57, no.1,  pp. 397--402, Jan. 2011.

\bibitem{Mc04} G. McGuire, ``On three weights in cyclic codes with two zeros,"
{\em Finite Fields Appl.,} vol. 10, no. 1,  pp. 97--104, Jan. 2004.

\bibitem{Rosen} P. Rosendahl, {\em Niho Type Cross-Correlation
Functions and Related
Equations,} TUCS Dissertations
No 53, August 2004.

\bibitem{veg} G. Vega, ``The weight distribution of an extened class of reducible cyclic codes,"
{\em IEEE Trans. Inform. Theory,} vol. 58, no. 7,  pp. 4862--4869, July 2012.

\bibitem{veg2} G. Vega and C. A. V\'azquez, ``The weight distribution of a family of reducible
cyclic codes," in: {\em Arithmetic of Finite Fields,} Lecture Notes in Computer Science 7369,
Springer-Verlag, 2012, pp. 16--28.

\bibitem{WT} B. Wang, C. Tang, Y. Qi, Y. Yang, and M. Xu, ``The weight distributions of cyclic codes
and elliptic curves," {\em IEEE Trans. Inform. Theory,} vol. 58, no. 12, pp. 7253--7259, Dec. 2012.

\bibitem{Xia} Y. Xia, X. Zeng, and L. Hu, ``Further crosscorrelation properties of sequences
with the decimation factor $d = (p^n+1)/(p+1) +􀀀(p^n􀀀-1)/2$,"
{\em Appl. Algebra Eng. Commun. Comput.}, vol. 21, no. 5,  pp. 329--342, Nov. 2010.

\bibitem{X} M. Xiong, ``The weight distributions of a class of cyclic codes," {\em Finite Fields Appl.,}
vol. 18, no. 5, pp. 933--945, Sept. 2012.

\bibitem{X2} M. Xiong, ``The weight distributions of a class of cyclic codes II,"
{\em Des. Codes Cryptogr.,} DOI 10.1007/s10623-012-9785-0.

\bibitem{YCD} J. Yuan, C. Carlet, and C. Ding, ``The weight distribution of a class of linear codes
from perfect nonlinear functions," {\em IEEE Trans. Inform. Theory,} vol. 52, no. 2,  pp. 712--717, Feb. 2006.

\bibitem{YD06} J. Yuan and C. Ding, ``Secret sharing schemes from three classes of linear codes,"
{\em IEEE Trans. Inform. Theory,} vol. 52, no.1,  pp. 206--212, Jan. 2006.

\bibitem{ZD2} Z. Zhou and C. Ding, ``Seven classes of three-weight cyclic codes,'' {\em IEEE Trans. Commun.}, 
accepted for publication.  Online: http://ieeexplore.ieee.org/stamp/stamp.jsp?tp=\&arnumber=6567875.


\end{thebibliography}
\end{document}